\documentclass{ws-mpla}
\usepackage{url}
\usepackage{hyphenat}
\usepackage[super]{cite}
\usepackage{graphicx}
\newcommand{\xv}{{\vec{\rm x}}}
\newcommand{\ka}{{\vec{\rm k}}}

\newcommand{\q}{{\vec{\rm q}}}
\newcommand{\p}{{\vec{\rm p}}}
\newcommand{\vs}{{\vec{\rm v}}}
\newcommand{\n}{{\vec{\rm n}}}

\newcommand{\rr}{{\vec{\rho }}}

\newcommand{\Xv}{{\vec{\rm X}}}

\def\bldmth#1{%
\mathchoice
{{\hbox{\boldmath$\displaystyle#1$\unboldmath}}}%
{{\hbox{\boldmath$\textstyle#1$\unboldmath}}}%                                                                                        
{{\hbox{\boldmath$\scriptstyle#1$\unboldmath}}}%                                                                                      
{{\hbox{\boldmath$\scriptscriptstyle#1$\unboldmath}}}%                                                                                
}
\def\vec#1{\bldmth{#1}}
\begin{document}

\markboth{S.~E.~Korenblit and D.~V.~Taychenachev}
{Extension of Grimus-Stockinger formula from 
 operator expansion of free Green function}

%%%%%%%%%%%%%%%%%%%%% Publisher's Area please ignore %%%%%%%%%%%%%%
\catchline{}{}{}{}{}
%%%%%%%%%%%%%%%%%%%%%%%%%%%%%%%%%%%%%%%%%%%%%%%%%%%%%%%%%%%%%%%%%%%

\title{EXTENSION OF GRIMUS-STOCKINGER FORMULA FROM \\
 OPERATOR EXPANSION OF FREE GREEN FUNCTION}

\author{\footnotesize S.~E.~Korenblit and D.~V.~Taychenachev}

\address{
Department of Physics, Irkutsk State University, \\
Gagarin blvd 20, Irkutsk 664003, Russia \\
korenb@ic.isu.ru
}

\maketitle

\pub{Received (Day Month Year)}{Revised (Day Month Year)}

\begin{abstract}
\indent The operator expansion of free Green function of Helmholtz  
equation for arbitrary $N$- dimension space leads to asymptotic extension of 
3- dimension Grimus-Stockinger formula closely related to multipole expansion. 
Analytical examples inspired by neutrino oscillation and neutrino deficit problems 
are considered for relevant class of wave packets.

\end{abstract}

\ccode{PACS Nos.: 03.65.Db, 11.80.Fv, 14.60.Pq}
%%%%%%%%%%%%%%%%%%%%%%%%%%%%%%%%%%%%%%%%%%%%%%%%%%%%%%%%%%%%%%%%%%%%%%%%%%%

\section{Introduction}

A considerable efforts was made recently\cite{k_t,N_shk,NN_shk} to extend 
so called Grimus-Stockinger theorem,\cite{Grimus} which is the main tool of 
the modern theory of neutrino oscillations\cite{Beuthe,NN,ahm} and gives the 
leading asymptotic behavior with $R=|\vec{\rm R}| \rightarrow \infty$ for the integral:
\begin{equation}
{\cal J}(\vec{\rm R}) = \int \frac{d^3 q}{(2\pi)^3}
\frac{e^{-i(\q\cdot\vec{\rm R})}\Phi(\q)}{(\q^2 - k^2 - i0)}\sim
\frac{e^{i{k}R}}{4\pi R}\Phi\left(-k\n\right)\left[1+ O(R^{-1/2})\right], 
\label{1}
\end{equation}
where: $\vec{\rm R} = R\n$, $\n^2=1$, and the function $\Phi(\q) \in C^3$ 
decreases at least like $1/\q^2$ together with its first and second derivatives. 
In the next section of the present letter it is reminded that the possibility of further 
asymptotic expansion and the order of leading correction depend\cite{olv} on 
the chosen properties\cite{k_t,N_shk} of this function. 
In section 3 for appropriate space of $\Phi(\q)$ a closed formula and simple recurrent 
relation for coefficients of asymptotic expansion of ${\cal J}(\vec{\rm R})$ in all 
orders of $R^{-s}$ are obtained. 
Such extension has particular importance for explanation \cite{N_shk,NN_shk} of 
observed \cite{anom} (anti-) neutrino deficit at short distances from the sources 
discussed in section 4. 
But it may find much more wide implementation in quantum physics and optics, when the 
Green function of Helmholtz equation is used. 
In sections 4, 5 the space chosen for the functions $\Phi(\q)$ is advocated 
on the ground of quantum field theory of wave packets.\cite{NN,k_t2} %%% motivated 
Some useful relations and generalization onto $N$- dimension case are placed in 
Appendix.

\section{Preliminaries}

In order to understand the physical nature of asymptotic expansion, we notice that for 
infinitely differentiable $\Phi(\q)$ uniquely representable by its Taylor expansion 
for any finite $|\q|<\infty$: $e^{-i(\q\cdot\vec{\rm R})}\Phi(\q)=
\Phi\left(i\vec{\nabla}_{\vec{\rm R}}\right)e^{-i(\q\cdot\vec{\rm R})}$, and then, 
formally: 
\begin{equation}
{\cal J}(\vec{\rm R})=
\Phi\left(i\vec{\nabla}_{\vec{\rm R}}\right)\frac{e^{ik R}}{4\pi R}=
\Phi\left(-i\vec{\nabla}_{\xv}\right)
\frac{e^{ik|\vec{\rm R}-\xv|}}{4\pi|\vec{\rm R}-\xv|}\biggr|_{\xv=0},
\label{2}
\end{equation}
where the differential vector operator in spherical basis 
$\n,\vec{\eta}_\vartheta,\vec{\eta}_\varphi$ has the following properties: 
\begin{eqnarray}
&&\!\!\!\!\!\!\!\!\!\!\!\!\!\!\!\!\!\!\!\!\!\!\!\! 
\n=(\sin\vartheta\cos\varphi,\sin\vartheta\sin\varphi,\cos\vartheta),\quad 
\vec{\eta}_\vartheta=\partial_\vartheta\n,\;\;\;
\sin\vartheta \vec{\eta}_\varphi=\partial_\varphi\n, 
\label{4} \\
&&\!\!\!\!\!\!\!\!\!\!\!\!\!\!\!\!\!\!\!\!\!\!\!\! 
\vec{\nabla}_{\vec{\rm R}}=\n\partial_R+\frac 1R \vec{\partial}_{\n}, \quad 
(\n\cdot\vec{\nabla}_{\vec{\rm R}})=\partial_R, \quad 
\vec{\partial}_{\n}=\vec{\eta}_\vartheta\partial_\vartheta+
\frac{\vec{\eta}_\varphi}{\sin\vartheta}\partial_\varphi, 
\label{3} \\
&&\!\!\!\!\!\!\!\!\!\!\!\!\!\!\!\!\!\!\!\!\!\!\!\! 
(\n\cdot\vec{\partial}_{\n})=0,\quad 
(\vec{\partial}_{\n}\cdot\n)=2,\quad 
(\n\times\vec{\partial}_{\n})^2=\vec{\partial}^2_{\n},\quad 
(\n\times\vec{\partial}_{\n})=i\vec{L}_{\n},
\label{5} \\
&&\!\!\!\!\!\!\!\!\!\!\!\!\!\!\!\!\!\!\!\!\!\!\!\! 
-\vec{\partial}^2_{\n}=\vec{L}^2_{\n}=
2R(\n\cdot\vec{\nabla}_{\vec{\rm R}})+R^2\left(
(\n\cdot\vec{\nabla}_{\vec{\rm R}})^2-\vec{\nabla}^2_{\vec{\rm R}}\right),\;
\mbox{ whence,}
\label{6} \\
&&\!\!\!\!\!\!\!\!\!\!\!\!\!\!\!\!\!\!\!\!\!\!\!\! 
\mbox{for }\;\cos\vartheta=c:\;\;
\vec{L}^2_{\n}=
-\left[\partial_c(1-c^2)\partial_c+(1-c^2)^{-1}\partial^2_\varphi\right]
\equiv {\cal L}_{\n},
\label{7}
\end{eqnarray}
and the well known representation at point $\vec{\rm R}$ for the spherical wave coming 
from point $\xv$, as a free Schr\"odinger's 3-dimensional Green function,\cite{T} is 
used (see (\ref{A_4})): 
\begin{equation}
\frac{e^{ik|\vec{\rm R}-\xv|} }{4\pi|\vec{\rm R}-\xv|}=\!
\int\!\frac{d^3{\rm q}}{(2\pi)^3}\frac{e^{\pm i(\q\cdot(\vec{\rm R}-\xv))}}
{(\q^2-k^2-i0)},
\label{8}
\end{equation}
which for $\xv=0$ satisfys to the equation: 
\begin{equation}
\left(-\vec{\nabla}^2_{\vec{\rm R}}-k^2\right)
\frac{e^{\pm ik R}}{4\pi R}=\delta_3(\vec{\rm R}), 
\label{9}
\end{equation}
where for $R>0$ $\vec{\nabla}^2_{\vec{\rm R}}$ is given also by (\ref{A_2}) with $N=3$.  
Since for $R>0$ the right hand side of this equation is zero, $\delta_3(\vec{\rm R})=0$, 
for spherically symmetric case one immediately obtains from (\ref{2}) formally an exact 
answer, which really takes place at least for the function $\Psi(k^2)$ regular and 
bounded in upper half plane ${\rm Im}\,k\geq 0$ of complex variable $k$:  
\begin{equation}
\mbox{when }\;\Phi(\q) = \Psi(q^2),\;\mbox{ for }\;\q^2=q^2, \;\mbox{ then: }\;
{\cal J}(\vec{\rm R}) = \Psi(k^2) e^{ikR}/(4\pi R). 
\label{10}
\end{equation}
Therefore the higher order corrections originate only by asymmetry of function 
$\Phi(\q)$ relative to the directions of $\q$ in accordance with our previous 
result.\cite{k_t} 
Indeed, in order to obtain them, we supposed\cite{k_t} that $\Phi(\q)$ and 
its first and second derivatives are represented by Fourier-transforms as: 
\begin{equation}
\Phi(\q)=\!\int\! d^3{\rm x}\,e^{i(\q\cdot\xv)}\,\phi(\xv), \quad 
\vec{\nabla}_q\Phi(\q)=i\!\int\! d^3{\rm x}\,e^{i(\q\cdot\xv)}\,\xv\,\phi(\xv), 
\mbox{ and so on.}
\label{11}
\end{equation}
Then by interchanging the order of integration and using the Eq. (\ref{8}) we found 
the representation: 
\begin{equation}
{\cal J}(\vec{\rm R})=\!
\int\! d^3{\rm x}\,\frac{e^{ik|\vec{\rm R}-\xv|}}{4\pi|\vec{\rm R}-\xv|}\,\phi(\xv). 
\label{12}
\end{equation}
Substituting here the expansion, which in the exponential should always contain one 
additional order in compare with the number of orders in denominator, for:  
\begin{eqnarray}
&&\!\!\!\!\!\!\!\!\!\!\!\!\!\!\!\!\!
|\vec{\rm R}-\xv|=R\left[1-2\frac{(\n\cdot\xv)}{R}+\frac{\xv^2}{R^2}\right]^{1/2},\; 
%%%% = R\left[\left(\!1-\frac{(\n\cdot\xv)}{R}\!\right)^2\!\!+
%%%% \frac{(\n\times\xv)^2}{R^2}\right]^{1/2}, 
\;\mbox{ whih: }\; \xv^2-(\n\cdot\xv)^2 \leftrightharpoons (\n\times\xv)^2,
\label{13_0} \\
&&\!\!\!\!\!\!\!\!\!\!\!\!\!\!\!\!\!
|\vec{\rm R}-\xv|=R-(\n\cdot\xv)+\frac{\xv^2-(\n\cdot\xv)^2}{2R}+\ldots,
\label{13}
\end{eqnarray}
we had the corresponding expansion of integral (\ref{12}) up to $O(R^{-3})$:
\[
{\cal J}(\vec{\rm R})=\frac{e^{ik R}}{4\pi R}
\int d^3{\rm x}\,e^{-ik(\n\cdot\xv)}\phi(\xv)
\left[1+\frac{(\n\cdot\xv)}{R}+\frac{ik}{2R}\left(\xv^2-(\n\cdot\xv)^2\!\right)+\ldots 
\right],
\]
that by making use of (\ref{11}) was immediately transcribed as:
\begin{eqnarray}
&&\!\!\!\!\!\!\!\!\!\!\!\!\!\!\!\!\!
{\cal J}(\vec{\rm R})=\frac{e^{ik R}}{4\pi R}
\left[1-\frac i{R} (\n\cdot\vec{\nabla}_q)+
\frac{ik}{2R}\left((\n\cdot\vec{\nabla}_q)^2-\vec{\nabla}^2_q\right)+\ldots\right]
\Phi(\q)\biggr|_{\q=-k\n},
\label{14}  \\
&&\!\!\!\!\!\!\!\!\!\!\!\!\!\!\!\!\!
\mbox{with: }\;(\n\cdot\vec{\nabla}_q)\Phi(\q)\bigr|_{\q=-k\n}=
-(\n\cdot\vec{\nabla}_\ka)\Phi(-\ka)=-\partial_k\Phi(-k\n), \;\mbox{ and so on.}
\label{15} 
\end{eqnarray}
The next observation is that due to (\ref{15}), with $\ka=k\n$, for $\q = -\ka$, 
$\vec{\nabla}_{\q}= -\vec{\nabla}_{\ka}$, from (\ref{6}), (\ref{7}) with 
$\vec{\rm R}\mapsto\ka$, one has:   
$2k(\n\cdot\vec{\nabla}_{\ka})+k^2\left((\n\cdot\vec{\nabla}_{\ka})^2-
\vec{\nabla}^2_{\ka}\right)={\cal L}_{\n}$, and expression (\ref{14}) takes the following 
simple form: 
\begin{eqnarray}
&&\!\!\!\!\!\!\!\!\!\!\!\!\!\!\!\!\!
{\cal J}(\vec{\rm R})=\frac{e^{ik R}}{4\pi R}
\left[1+\frac i{R} (\n\cdot\vec{\nabla}_{\ka})+\frac{i k}{2R}
\left((\n\cdot\vec{\nabla}_{\ka})^2-\vec{\nabla}^2_{\ka}\right)+\ldots\right]\Phi(-\ka)=
\label{16}\\
&&\!\!\!\!\!\!\!\!\!\!\!\!\!\!\!\!\!
=\frac{e^{ik R}}{4\pi R}\left[1+\frac i{2kR}{\cal L}_{\n}+O(R^{-2})\right]\Phi(-k\n).
\label{17}
\end{eqnarray}
In the next section by using the another surprisingly simple way, this result is 
generalized onto all orders of $R^{-s}$. 
It should be noted, that in spite of the one and the same final results of the both 
Eqs. (\ref{14}) and (\ref{16}), (\ref{17}), the substitution like (\ref{13_0}): 
$\vec{\nabla}^2_q-(\n\cdot\vec{\nabla}_q)^2\leftrightharpoons (\n\times\vec{\nabla}_q)^2$  
is correct till $\q\neq -k\n$ for expressions\cite{k_t,N_shk,NN_shk} like (\ref{14}), but it 
becomes incorrect for non-commutative operators $\n$ and $\vec{\nabla}_{\ka}$ in 
expressions like (\ref{16}). This difference from Ref. 1-3 %%%\cite{k_t,N_shk,NN_shk} 
 below eventually gives rise an explicit operator expression for higher coefficients of 
this asymptotic expansion with arbitrary $s$.
%%%%%%%%%%%%%%%%%%%%%%%%%%%%%%%%%%%%%%%%%%%%%%%%%%
\section{Asymptotic and multipole expansions}
%%%%%%%%%%%%%%%%%%%%%%%%%%%%%%%%%%%%%%%%%%%%%%%%%

\begin {lemma} 
For $\vec{\rm R}=R\n$, $\xv=r\vs$, 
$\vs=(\sin\beta\cos\alpha, \sin\beta \sin\alpha,\cos\beta)$, 
$|\xv|=r<R$, with operator ${\cal L}_{\n}=\vec{L}^2_{\n}$ 
(or $\n\mapsto\vs$) defined by Eqs. (\ref{5})--(\ref{7}) and positively defined operator 
${\cal L}_{\n}+\frac 14=(\Lambda_{\n}+\frac 12)^2$, so that 
$\Lambda_{\n}+\frac 12=\sqrt{{\cal L}_{\n}+\frac 14}$ is positively defined: 
\begin{eqnarray}
&&\!\!\!\!\!\!\!\!\!\!\!\!\!\!\!\!\!\!\!\!\!\!\!\! 
\frac{e^{ik|\vec{\rm R}-\xv|}}{4\pi|\vec{\rm R}-\xv|}=
\frac{\chi_{\Lambda_{\n}}(-ikR)}{4\pi R}e^{-ik(\n\cdot\xv)}\sim
\label{18_0} \\
&&\!\!\!\!\!\!\!\!\!\!\!\!\!\!\!\!\!\!\!\!\!\!\!\! 
\sim\frac{e^{ikR}}{4\pi R} \left\{1+\sum^\infty_{s=1}\frac{ \displaystyle 
\prod\limits^s_{\mu=1}\left[{\cal L}_{\n}-\mu(\mu-1)\right]}
{s!(-2ikR)^s}\right\} e^{-ik(\n\cdot\xv)}. 
\label{18} 
\end{eqnarray}
\end {lemma}
\begin {proof} %% {\sl Proof:} 
The expression (\ref{18_0}) is a formal operator rewrite for $R>r$ of 
the usual multipole expansion of the Green function\cite{T} (\ref{8}) via the 
corresponding expansion of the plane wave,\cite{T} given also by formulas  
(8.533), (8.534) of Ref. 12, %%%\cite{b_e2},
 with the formally introduced instead $l$, but not 
really appeared operator $l\mapsto\Lambda_{\n}$: 
\begin{eqnarray}
&&\!\!\!\!\!\!\!\!\!\!\!\!\!\!\!\!\!\!\!\!\!\!\!\! 
\frac{e^{\pm ik|\vec{\rm R}-\xv|}}{4\pi|\vec{\rm R}-\xv|}=
\frac{1}{kRr}\sum^\infty_{l=0}i^{\mp l}\chi_l(\mp ikR)\,\psi_{l\,0}(kr)
\sum^l_{m=-l}Y^m_l(\n)\overset{*}{Y}\!{}^m_l(\vs),
\label{19} \\
&&\!\!\!\!\!\!\!\!\!\!\!\!\!\!\!\!\!\!\!\!\!\!\!\! 
e^{\mp i(\ka\cdot\xv)}=
\frac{4\pi}{kr}\sum^\infty_{l=0}i^{\mp l}\,\psi_{l\,0}(kr)
\sum^l_{m=-l}Y^m_l(\n)\overset{*}{Y}\!{}^m_l(\vs).
\label{20} 
\end{eqnarray}
Here the spherical functions $Y^m_l(\n)=\langle\n|l\,m\rangle$ and Legendre 
polynomials $P_l(\xi)$ for $\xi=c=\cos\vartheta$ or $\xi=(\n\cdot\vs)$,
as eigenfunctions of self-adjoint operator (\ref{6}), (\ref{7}) on the unit sphere,  
satisfy to well known orthogonality, parity and completeness 
conditions\cite{T,b_e2,vil,b_l} with 
delta-function $\delta_\Omega(\n\,,\vs)$ on the unit sphere: 
\begin{eqnarray}
&&\!\!\!\!\!\!\!\!\!\!\!\!\!\!\!\!\!\!\!\!\!\!\!\! 
\vec{L}^2_{\n}Y^m_l(\n)=l(l+1)Y^m_l(\n), \quad \quad
\vec{L}^2_{\n}P_l(\xi)=l(l+1)P_l(\xi),
\label{21} \\
&&\!\!\!\!\!\!\!\!\!\!\!\!\!\!\!\!\!\!\!\!\!\!\!\! 
\int d\Omega(\n) \overset{*}{Y}\!{}^m_l(\n)Y^{m^\prime}_j(\n)=
\delta_{lj}\delta_{mm^\prime}, \quad 
\sum^\infty_{l=0}\frac{(2l+1)}{4\pi}P_l\left((\n\cdot\vs)\right)=
\delta_\Omega(\n\,,\vs), 
\label{22_0} \\
&&\!\!\!\!\!\!\!\!\!\!\!\!\!\!\!\!\!\!\!\!\!\!\!\!  
\sum^l_{m=-l}Y^m_l(\n)\overset{*}{Y}\!{}^m_l(\vs)\equiv 
\frac{(2l+1)}{4\pi} P_l\left((\n\cdot\vs)\right), \quad \; (-1)^lY^m_l(\n)=Y^m_l(-\n). 
\label{22} 
\end{eqnarray}
The solutions $\chi_l(\mp ikr)$, $\psi_{l\,0}(kr)$ of free radial Schr\"oedinger   
equation:  
\begin{equation}
\left[r^2\left(\frac 1r\,\partial^2_r\,r+k^2\right)\right]\frac{\psi_{l\,0}(kr)}{r}=
l(l+1)\frac{\psi_{l\,0}(kr)}{r},  
\label{26} 
\end {equation}
are defined by Macdonald $K_\lambda(z)$ and Bessel $J_\lambda(y)$ 
functions,\cite{olv,T,b_e2,vil,b_l} that for integer $l$ i.e. half integer 
$\lambda=l+\frac 12$ are reduced to elementary functions:\cite{olv,b_e2}  
\begin{eqnarray}
&&\!\!\!\!\!\!\!\!\!\!\!\!\! 
\chi_l(b R)\equiv \left(\frac{2bR}{\pi}\right)^{1/2}\!\!
K_{l+{\frac 12}}(bR), \quad 
\chi_l(b R) \biggr|_{l=\mathrm{int}}\!\!\! =
%%% \underset{l=\mathrm{int}}{\Longrightarrow}
e^{-b R}\sum^l_{s=0}\frac{(l+s)!}{s!(l-s)!(2b R)^s},
\label{24} \\
&&\!\!\!\!\!\!\!\!\!\!\!\!\! 
\psi_{l\,0}(kr) \equiv \left(\frac{\pi kr}{2}\right)^{1/2}\!\!\!J_{l+{\frac 12}}(kr)
\equiv \frac{1}{2i}\left[i^{-l}\chi_l(-ikr)-i^l\chi_l(ikr)\right].
\label{24_1} 
\end{eqnarray}
The function $K_\lambda(z)$ (\ref{A_1}) is a whole function\cite{olv,b_e2} of 
$\lambda^2$, that is the reason, why the above introduced in (\ref{18_0}) well 
defined operator $\Lambda_{\n}$ does not appear explicitly. 

The expression (\ref{18}) is the known asymptotic series of expression (\ref{18_0}) at 
$R\to\infty$, as an infinite asymptotic version\cite{olv,b_e2} of the sum (\ref{24}) for 
arbitrary non-integer $l$, $|{\rm arg}(bR)|<3\pi/2$. 
It directly results  also from the substitution of finite sums (\ref{22}), (\ref{24}) by 
interchanging the order of summations, converted the Eq. (\ref{19}) into the sum:
\begin{eqnarray}
&&\!\!\!\!\!\!\!\!\!\!\!\!\!\!\!\!\!\!\!\!\!\!\!\! 
%%% \frac{e^{ik|\vec{\rm R}-\xv|}}{4\pi|\vec{\rm R}-\xv|}=
\frac{e^{ikR}}{4\pi R}\sum^\infty_{s=0}\frac {1}{s!(-2ikR)^s}\,\frac{1}{kr} 
\sum^\infty_{l=s}i^{-l} \frac{(l+s)!}{(l-s)!}\,\psi_{l\,0}(kr)
(2l+1) P_l\left((\n\cdot\vs)\right), 
\label{27} \\
&&\!\!\!\!\!\!\!\!\!\!\!\!\!\!\!\!\!\!\!\!\!\!\!\! 
\mbox{where : }\;
\frac{(l+s)!}{(l-s)!}=\prod\limits^s_{\mu=1}(l-\mu+1)(l+\mu)=
\prod\limits^s_{\mu=1}\left[l(l+1)-\mu(\mu-1)\right],
\label{28}
\end{eqnarray}
equals to zero for all missing summands with $0\leq l\leq s-1$ automatically and due to  
Eqs. (\ref{21}) may be factored out from the sum over $l$ as operator product in the 
right hand side of Eq. (\ref{18}). The formal addition of all such missing in fact zero 
summands with $0\leq l\leq s-1$ completes then sum over $l$ into the plane wave (\ref{20}) 
and converts the expression (\ref{27}) into the expansion (\ref{18}).  \\ 
\noindent
{\sl Remark.} The operator ${\cal L}_{\n}$ in Eq. (\ref{18}) with the 
same success may be replaced by operator in square brackets of the left hand side of 
Eq. (\ref{26}) or by the same with interchanging $r\rightleftharpoons k$.  
\end {proof} %%% $\square$ 
%%% \noindent {\bf Theorem 1:} 
\begin {theorem}
Let $\Phi(\q)\in S(\vec{\rm R}^3_\q)$, the space of functions infinitely 
differentiable $\forall\, \q\in\vec{\rm R}^3_\q$, decrease faster than any power of 
$1/|\q|$ together with all its derivatives. Then integral ${\cal J}(\vec{\rm R})$ 
(\ref{1}), (\ref{12}) at $R\to\infty$ admits asymptotic expansion, which has asymptotic 
sense\cite{olv} even though $\phi(\xv)$ in Eq. (\ref{11}) has a finite support: 
\begin{eqnarray}
&&\!\!\!\!\!\!\!\!\!\!\!\!\!\!\!\!\!\!\!\!\!\!\!\! 
{\cal J}(\vec{\rm R})\sim \frac{e^{ikR}}{4\pi R} \Phi(-k\n) \left\{
1+\sum^\infty_{s=1}\frac{C_s(k,\n)}{(-2ikR)^s}\right\},\; \mbox{ with:}
\label{29} \\
&&\!\!\!\!\!\!\!\!\!\!\!\!\!\!\!\!\!\!\!\!\!\!\!\! 
\Phi(-k\n) C_s(k,\n)=
\frac{1}{s!}\prod\limits^s_{\mu=1}\left[{\cal L}_{\n}-\mu(\mu-1)\right]\Phi(-k\n), 
\;\mbox{ or:}
\label{30} \\
&&\!\!\!\!\!\!\!\!\!\!\!\!\!\!\!\!\!\!\!\!\!\!\!\! 
\Phi(-k\n) C_s(k,\n)=\frac{{\cal L}_{\n}-s(s-1)}{s}\Phi(-k\n)C_{s-1}(k,\n),
\quad   C_0(k,\n)=1,
\label{30_0}
\end{eqnarray}
and which is equivalent to infinite reordering of its multipole expansion:
\begin{eqnarray}
&&\!\!\!\!\!\!\!\!\!\!\!\!\!\!\!\!\!\!\!\!\!\!\!\! 
{\cal J}(\vec{\rm R})\sim \frac{1}{4\pi R}
\sum^\infty_{j=0}\chi_j(-ikR)\sum^j_{m=-j}B^m_j(k) Y^m_j(\n), 
\label{31} \\
&&\!\!\!\!\!\!\!\!\!\!\!\!\!\!\!\!\!\!\!\!\!\!\!\! 
\mbox{with: }\;
\Phi(-k\n) C_s(k,\n)=
\frac{1}{s!}\sum^\infty_{j=s}\frac{(j+s)!}{(j-s)!}\sum^j_{m=-j}\!\!B^m_j(k) Y^m_j(\n), 
\label{32} \\
&&\!\!\!\!\!\!\!\!\!\!\!\!\!\!\!\!\!\!\!\!\!\!\!\! 
\mbox{for: }\;
\Phi(-k\n)=\sum^\infty_{j=0}\sum^j_{m=-j}B^m_j(k) Y^m_j(\n), \quad 
\phi(\xv)=\!\int\!\frac{d^3{\rm q}}{(2\pi)^3}\Phi(\q)\,e^{-i(\q\cdot\xv)}.
\label{32_1}
\end{eqnarray}
\end {theorem}
\begin {proof} 
Since Fourier transformation (\ref{11}) maps the space $S(\vec{\rm R}^3)$ 
into itself\cite{blt} the function $\phi(\xv)\in S(\vec{\rm R}^3_\xv)$ also and is 
represented by the inverse Fourier transform (\ref{32_1}). Let suppose at first that 
$\phi(\xv)$ has a finite support at $|\xv|\leq r_0$. Then for $R>r_0$ we can directly 
substitute the expressions (\ref{18}) into representation (\ref{12}) of 
${\cal J}(\vec{\rm R})$ with the following result after interchange of the order of 
summation, differentiation and integration of Fourier transform (\ref{11}), 
justified\cite{olv} also for the asymptotic series: 
\begin{equation}
{\cal J}(\vec{\rm R})=
\frac{{\chi}_{\Lambda_{\n}}(-ikR)}{4\pi R}\,\Phi(-k\n)\sim 
\frac{e^{ikR}}{4\pi R}\left\{1+\sum^\infty_{s=1}\frac{ \displaystyle 
\prod\limits^s_{\mu=1}\left[{\cal L}_{\n}-\mu(\mu-1)\right]}
{s!(-2ikR)^s}\right\} \Phi(-k\n).
\label{33}
\end{equation}
This is exactly asymptotic expansion (\ref{29}) with dimensionless coefficients 
$C_s(k,\n)$, defined by Eq. (\ref{30}). However that is not the case for the function 
$\phi(\xv)$ with infinite support. 
Estimating it for $r>R$ as $|\phi(\xv)|<{\rm C}_{\rm M}/r^{\rm M}$ with arbitrary finite 
${\rm M}\gg 1$, the two pieces of correction that should be added, are easy estimating 
as:  
\begin{eqnarray}
&&\!\!\!\!\!\!\!\!\!\!\!\!\!\!\!\!\!
\Delta_R{\cal J}=\!\!\!
\int\limits_{r>R}\!\!\! d^3{\rm x}\frac{e^{ik|\vec{\rm R}-\xv|}}{4\pi|\vec{\rm R}-\xv|}\,
\phi(\xv),\quad
\Delta_R\Phi=-\frac{{\chi}_{\Lambda_{\n}}(-ikR)}{4\pi R}\!\!\!
\int\limits_{r>R}\!\!\!d^3{\rm x}\,e^{-ik(\n\cdot\xv)}\phi(\xv),
\label{34} \\
&&\!\!\!\!\!\!\!\!\!\!\!\!\!\!\!\!\!
|\Delta_R{\cal J}|<\frac{{\rm C}_{\rm M}}{({\rm M}-2)R^{{\rm M}-2}}, \quad \;\;
|\Delta_R\Phi|<\frac{{\rm C}_{\rm M}}{({\rm M}-3)R^{{\rm M}-2}}\left[1+ O(R^{-1})\right].
\label{34_1} 
\end{eqnarray}
Due to this corrections and arbitrariness of ${\rm M}\gg 1$, the expansion (\ref{29}), 
(\ref{33}) acquires an additional asymptotic meaning\cite{olv} in comparison with 
expansion (\ref{18}). The terms with $s=0,1$ reproduce the results\cite{k_t} 
(\ref{10}), (\ref{17}) of previous section. 

Coefficients of multipole expansion in (\ref{32_1}) may be defined by the following two 
different ways, whose equivalence, for $\xv=r\vs$:   
\begin{equation}
B^m_l(k)\equiv\! \int\! d\Omega(\n)\, \Phi(-k\n)\,\overset{*}{Y}\!{}^m_l(\n)=
4\pi \!\int \! d^3{\rm x}\,\frac{\psi_{l\,0}(kr)}{kr}\, i^{-l} \,
\overset{*}{Y}\!{}^m_l(\vs)\,\phi(\xv), 
\label{35} 
\end{equation}
follows from Eqs. (\ref{11}), (\ref{20})--(\ref{22}) and (\ref{24_0}), with inverse 
Fourier transform (\ref{32_1}).   

Substitution of multipole expansion of Green function (\ref{19}) into the integral 
(\ref{12}) by making use of definition (\ref{35}), after the same steps and under the 
same conditions for $\phi(\xv)$ as above leads to the {\bf multipole} expansion 
(\ref{31}), which with the help of expansion (\ref{24}) transcribes again as asymptotic 
expansion (\ref{29}) with the coefficients given by Eq. (\ref{32}). But the same 
expression (\ref{32}) is obtained by direct substitution of (\ref{32_1}) into the 
definition (\ref{30}) by means of (\ref{21}), (\ref{28}). This confirms the equivalence 
of expansions (\ref{29}) and (\ref{31}). \\
\noindent
{\sl Corollary 1.}
From (\ref{29}) with $C_s=C_s(k,\n)$, for the squared absolute value follows: 
\begin{eqnarray}
&&\!\!\!\!\!\!\!\!\!\!\!\!\!\!\!\!\!
(4\pi R)^2\left|{\cal J}(\vec{\rm R})\right|^2\sim
\nonumber \\
&&\!\!\!\!\!\!\!\!\!\!\!\!\!\!\!\!\! 
\sim|\Phi(-k\n)|^2\left\{1+\sum^\infty_{s=1}
\frac{i^s[C_s+(-1)^s\overset{*}{C}_s]}{(2kR)^s}+
\sum^\infty_{\zeta=2}\frac{i^\zeta}{(2kR)^\zeta}
\sum^{\zeta-1}_{s=1}(-1)^{\zeta-s}C_s\overset{*}{C}_{\zeta-s}\right\}.
\label{37}  
\end{eqnarray}
For $\Phi(-\ka)=\overset{*}{\Phi}(-\ka)$: $C_s= \overset{*}{C}_s$, so $s\mapsto 2n$ 
at the first sum over $s$. But the second internal sum over $1\leq s\leq\zeta-1$ in 
(\ref{37}) is equal to itself with the multiplier $(-1)^\zeta$ and thus 
$\zeta\mapsto 2n$ also. Dividing further this sum over $1\leq s\leq 2n-1$ into two 
parts: with $1\leq s\leq n$ and $n+1\leq s\leq 2n-1$, and puting for the second sum 
$s=2n-s^\prime$, one finds: 
\begin{eqnarray}
&&\!\!\!\!\!\!\!\!\!\!\!\!\!\!\!\!\!\!\!\! 
(4\pi R)^2\left|{\cal J}(\vec{\rm R})\right|^2 \sim
\Phi^2(-k\n)\left\{1+\sum^\infty_{n=1}\frac{\Upsilon_n}{(2kR)^{2n}}\right\}, 
\;\mbox{ with: }\; \Upsilon_n=\Upsilon_n(k,\n), 
\label{38} \\
&&\!\!\!\!\!\!\!\!\!\!\!\!\!\!\!\!\!\!\!\!  
\Upsilon_n=(-1)^n\left[2C_{2n}+\sum^{2n-1}_{s=1}(-1)^{s}C_sC_{2n-s}\right]
=(-1)^n \sum\limits^{2n}_{s=0}(-1)^{s}C_sC_{2n-s}, 
\label{39}  \\
&&\!\!\!\!\!\!\!\!\!\!\!\!\!\!\!\!\!\!\!\! 
\mbox{or: }\;
\Upsilon_n=
2(-1)^n\sum\limits^{n}_{s=0}(-1)^{s}C_s C_{2n-s}-(C_n)^2, \quad 
\Upsilon_1=(C_1)^2-2C_2,  
\label{39_0} 
\end{eqnarray}
what, due to ``homogeneity'' over $C_s$ relative to index $s$ ``like a power'', 
coincides with the relations\cite{N_shk} for dimensional analogs of coefficients 
$C_s$ and $\Upsilon_s$. \\
\noindent
{\sl Corollary 2.} Asymtotic expansion of the integral 
$\widetilde{\cal J}(\vec{\rm R})={\cal J}(-\vec{\rm R})$ (\ref{1}) is obtained by 
substitutions $\n\mapsto -\n$, ${\cal L}_{-\n}={\cal L}_{\n}$ (\ref{7}) in Eqs. 
(\ref{29})--(\ref{30_0}), (\ref{33}), (\ref{37})--(\ref{39_0}).\end {proof} %%% $\square$

The explicit Eqs. (\ref{29}), (\ref{30}) and recurrent relation (\ref{30_0}) may be 
compared with seminumerical calculations\cite{N_shk,NN_shk} of the same expansion, 
that for arbitrary $\Phi(\q)\in S(\vec{\rm R}^3_\q)$ have to give the same coefficients 
for the same orders $R^{-s}$. For the first correction this is easily seen from 
Eqs. (\ref{14})--(\ref{17}) and the Eq. (8) of Ref. 2 %%\cite{N_shk} 
respectively with substitution 
$\vec{\nabla}^2_q-(\n\cdot\vec{\nabla}_q)^2\leftrightharpoons (\n\times\vec{\nabla}_q)^2$.  
The same takes place for the next few corrections computed there. 
It is clear that any common multiplicative dependence of function $\Phi(\q)$ on 
$|\q|=k$ does not affect the coefficients $C_s(k,\n)$ (\ref{30}) and 
$\Upsilon_n(k,\n)$ (\ref{39}). 
But there are no ways to trace this property for arbitrary order $s$ with the only 
numerically calculable coefficients in Eqs. (41)--(46) of Ref. 2 %%\cite{N_shk}. 
An important advantage of the results (\ref{29})--(\ref{30_0}) for coefficients of 
asymptotic expansion is not only the explicit arbitrariness of their order $s$, 
but the elucidating, that in fact they are defined by dependence of $\Phi(\mp k\n)$ on 
the {\sl unit vector $\n$ only}. This essentially reduces the number of possible 
degrees of freedom and simplifies the further analytical consideration. Indeed, when 
$B^m_l(k)=0$ for $l>j$, the dependence (\ref{30_0}), (\ref{32}), (\ref{39_0}) of 
$\Upsilon_s$ on function $\Phi(\q)$ (\ref{32_1}) leads to termination of the asymptotic 
expansion: $C_s=0=\Upsilon_s$, for $s>j$, which is also clear from Eq. (\ref{31}). 
If one has:  
\begin{eqnarray}
&&\!\!\!\!\!\!\!\!\!\!\!\!\!\!\!\!\! 
\Phi(-k\n) =\!\! \sum^j_{m=-j}\!\! B^m_j(k) Y^m_j(\n),\;\mbox{ then: }\; 
C_n=\frac{(j+n)!}{n!(j-n)!},\;\;\Upsilon_1 = 2j(j+1)>0, 
\label{40} \\
&&\!\!\!\!\!\!\!\!\!\!\!\!\!\!\!\!\! 
\Upsilon_j=(C_j)^2,\quad 
\Upsilon_n=\frac{2(-1)^n}{(2n)!}\sum\limits^{n}_{s=0}(-1)^s{\rm C}^s_{2n}
\frac{(j+s)!(j+2n-s)!}{(j-s)!(j-2n+s)!}-(C_n)^2,
\label{40_0}
\end{eqnarray}
where $0\leq n\leq j$, ${\rm max}(0,2n-j)\leq s\leq n$ and ${\rm C}^s_{2n}$ are binomial 
coefficients.
%%%%%%%%%%%%%%%%%%%%%%%%%%%%%%%%%%%%%%%%%%%%%%
\section{Neutrino deficit and wave packets}
%%%%%%%%%%%%%%%%%%%%%%%%%%%%%%%%%%%%%%%%%%%%%

At diagrammatic treatment of neutrino 
oscillation\cite{N_shk,NN_shk,Grimus,Beuthe,NN,ahm} the function $\Phi(\q)$ in Eq. 
(\ref{1}) represents an overlap function ${\cal F}(q)={\cal F}_C(q){\cal F}_D(q)$ 
defined for $q^\mu=(q^0,\q)$ as a real-valued product of convolutions of the wave packets 
of external particles participating in the processes of (anti-) neutrino creation $\{C\}$ 
and detection $\{D\}$ at respective space-time points. 
According to the Feynman rules for respective tree amplitude\cite{Grimus,Beuthe,NN,ahm} 
one has to deal with asymptotic expansion of ${\cal J}(\vec{\rm R})$ (\ref{1}) for the 
case of antineutrino\cite{Grimus} or with expansion of $\widetilde{\cal J}(\vec{\rm R})$ 
for the case of neutrino,\cite{Beuthe} where the vector 
$\vec{\rm R}=R\n=\Xv_{D}-\Xv_{C}$ for both cases defines the macroscopic geometrical 
parameters of observation:\cite{Beuthe} distance $R$ and direction $\n$.   
The event rate is obtained to be proportional to $|{\cal J}(\vec{\rm R})|^2$ or 
$|\widetilde{\cal J}(\vec{\rm R})|^2$ respectively (see for detail Refs. 2,5,6,7)
%%% \cite{N_shk,Beuthe,NN,ahm}), 
that due to (\ref{38}) in the leading and next 
to leading orders reads: 
\begin{equation}
\propto 
\frac{\Phi^2(\mp k\n)}{R^2}\left\{1+\frac{\Upsilon_1(k,\pm\n)}{(2kR)^2}+\ldots\right\}.    
\label{41_00}
\end{equation}
Here the absolute value of momentum $k=\sqrt{q^2_0-m^2_j}$ is defined by neutrino mass 
$m_j$ and its mean energy $q^0=\mp Q^0$ near its most probable energies 
$Q^0_{C}\approx Q^0_{D}$, defined in turn\cite{Grimus,Beuthe,NN,ahm} 
by fixed parameters of respective external wave packets, such as their masses $m_a$, 
4-vectors of most probable energy-momentum $p_a$ and corresponding widths 
$\sigma_a\ll m_a$, because in the plane-wave limit, $\sigma_a\to 0$, $\forall\,a$  
the overlap function keeps exact energy-momentum conservation for both vertices, fully 
destroying the neutrino-oscillations pattern,\cite{Beuthe} 
${\cal F}(q)|_{\sigma_a\to 0}\propto\delta_4(q\pm Q_{C})\delta_4(q\pm Q_{D})$, 
where $Q_{C}$ becomes then a full 4-momentum incoming into (anti-) neutrino 
propagator, whereas $Q_{D}$ is full its outgoing 4-momentum. 
For $\sigma_a>0$ the function ${\cal F}(q)$ can provide thus only approximate  or 
``smeared'' conservation\cite{NN} of some mean (anti-) neutrino energy 
$Q^0\approx Q^0_{C}\approx Q^0_{D}$ and some mean (anti-) neutrino 3-momentum 
$\vec{\rm Q}\approx\vec{\rm Q}_{C}\approx\vec{\rm Q}_{D}$, whose 
explicit form depends on the chosen models of wave packets and the used 
approximations.\cite{Beuthe,NN,ahm} 
So, the function $\Phi(\mp k\n)={\cal F}(q^0,\q)|_{\q=\mp k\n}$ may be more or less 
sharply peaked\cite{Beuthe,NN} and varies rapidly: near the point 
$k\n\simeq\vec{\rm Q}=|\vec{\rm Q}|\rr$, for the more rough case\cite{Beuthe} (a); 
or near the two adjacent points $k\n\simeq\vec{\rm Q}_{C,D}=|\vec{\rm Q}_{C,D}|\rr_{C,D}$, 
where $\rr^2_{C,D}=1$, for the more precisely case\cite{NN} (b). For the case (a), up 
to fully unessential now multiplicative dependence on $k$ we remain with a 
function $\Phi(\mp k\n)\mapsto f(\xi)$ of the one dimensionless variable 
$\xi=(\rr\cdot\n)$ only. 
For the case (b) without full symmetry relative to vectors $\rr_{C,D}$ we remain  
with a function $\Phi(\mp k\n)\mapsto{\cal W}(\xi_{C},\xi_{D})$ of two independent 
variables $\xi_{C,D}=(\rr_{C,D}\cdot\n)$, which in Gassian case\cite{Beuthe,ahm} or in 
sharply peak approximation\cite{NN} always contains a function 
$\Phi(\mp k\n)\propto{\cal H}(\zeta)$ of the one dimensionless variable 
$\zeta=(\n \vec{\rm B}\n)>0$ with the dimensionless symmetrical positively defined 
tensor $\vec{\rm B}$ built on the components of these vectors, whose traceless part 
$\vec{\rm B}_0$ only is in fact necessary.  
Some technique with these functions is collected in Appendix B. 
%%%%%%%%%%%%%%%%%%%%%%

When $\Upsilon_1(k,\pm\n)<0$, the last multiplier in expression (\ref{41_00}) leads to 
suppression factor to the usual inverse square law for the event rate that reads as: 
\begin{equation}
\propto \frac{1}{R^2}\left\{1-\frac{\varrho^2_0}{R^2}\right\},\;\,\mbox{ where: }\;\,
\varrho^2_0=-\, \frac{\overline{\Upsilon}_1(k,\pm\n)}{(2k)^2}, 
\label{41_01}
\end{equation}
and the averaging $\overline{\Upsilon}_1$ relative to directions of $\rr$ is implied for 
the case (a), that in fact approximately puts $\rr=\n$, $\xi=1$. For the case (b) one 
has to average relative to the vectors $\rr_{C,D}$ independently, what makes the final 
result more model-dependent. According to Refs. 2,3 %%%\cite{N_shk,NN_shk} 
this suppression factor provided by the higher order corrections to Grimus-Stockinger 
formula (\ref{1}) naturally explains the 
observed\cite{anom} (anti-) neutrino deficit in reactor and others short-baseline 
experiments (see Refs. 2,3,9 %%%\cite{N_shk,NN_shk,anom} 
and the references therein). 
Such explanation does not require any ``new physics'' and perhaps\cite{N_shk} can be 
tested experimentally. 
%%%%%%%%%%%%%%%%%%%%%%

Eqs. (\ref{30}), (\ref{39_0}) show that negative sign of coefficient $\Upsilon_1$ 
originates by the most rapid variation $\Phi(\mp k\n)$ with vector $\n$. Exponential 
variations are directly generated by the above limiting properties of the overlap 
function and are naturally arising in popular models of wave packets.\cite{Beuthe,NN,ahm} 
Since beyond oscillation problem it is not meaningless to leave overlap function be 
defined by one wave packet only, any model of wave packet also has to reproduce the 
above limiting properties of ${\cal F}_{C,D}(q)$. 

In the spirit of Ref. 6 %%\cite{NN} 
the following form of wave packet in momentum 
representation was suggested\cite{k_t2} for $q_\mu \zeta^\mu_a\equiv(q\zeta_a)>0$, with 
time-like 4-vector $\zeta^\mu_a$ of mean quantum numbers of wave packet, such that 
$\zeta^0_a>0$, $\zeta_{a\mu} \zeta^\mu_a\equiv\zeta^2_a>0$:
\begin{equation}
\zeta_a(p_a,\sigma_a)=p_a g_1(m_a,\sigma_a)+s_a g_2(m_a,\sigma_a), 
\quad
\phi^\sigma(\q,\p_a)={N}_\sigma\!\left(m_a,\zeta^2_a\right) e^{-(q\zeta_a)}, 
\label{41} 
\end{equation}
where both the momentum 4-vectors $q^\mu=(E_{\rm q},\q)$, $p^\mu_a=(E_{{\rm p}_a},\p_a)$ 
are on mass shell: $E_{\rm q}=\sqrt{\q^2+m^2_a}>0$, and so on, $s_a$- is a spin 4-vector, 
$g_{1}\gg |g_{2}|$- are some real functions\cite{k_t2} of mass $m_a$ and width  
$\sigma_a$, and $\aleph(\tau) m^{-2}_a={N}_\sigma>0$, is normalization 
constant with fixed asymptotic behavior\cite{k_t2} as a function of dimensionless 
invariant variable $\tau=m_a\sqrt{\zeta^2_a(p_a,\sigma_a)}$, at 
$\tau\to\infty$ $(\sigma_a\to 0)$, and at $\tau\to 0$ $(\sigma_a\to\infty)$:
\begin{equation}
\aleph(\tau)\underset{\tau\to\infty}{\longmapsto} 2(2\pi)^{3/2}\tau^{3/2}e^\tau, \;
\mbox{ and }\;
\aleph(\tau)\underset{\tau\to 0}{\longmapsto} \aleph(0)>0.
\label{41_0}
\end{equation}
Coordinate representation of this wave packet with fixed center $x_a$ for scalar case 
satisfies Klein-Gordon equation and for $x^\mu=(t,\xv)$ is defined\cite{k_t2} by 
Wightman function analytically continued\cite{blt} into the same future tube 
($V^+$: $\zeta^0_a>0$, $\zeta^2_a>0$) as:
\begin{eqnarray}
&&\!\!\!\!\!\!\!\!\!\!\!\!\!\!\!\!\!\!\!\!\!\!\!\! 
F_{p_ax_a}(x)= e^{-i(p_ax_a)}\!\int\!\frac{d^3{\rm q}}{(2\pi)^3 2E_{\rm q}}\,
\phi^\sigma(\q,\p_a)e^{-i(q(x-x_a))}=
\nonumber \\
&&\!\!\!\!\!\!\!\!\!\!\!\!\!\!\!\!\!\!\!\!\!\!\!\! 
= (-i)e^{-i(p_ax_a)}{N}_\sigma D^-_{m_a}\!\left(x-x_a-i\zeta_a(p_a,\sigma_a)\right), 
\quad (g_2\equiv 0). 
\label{45}
\end{eqnarray}
This wave packet conforms with general requirements of quantum field theory\cite{blt} 
and due to (\ref{41_0}) admits adequate description\cite{k_t2} of both the limits to the 
states localized in momentum and coordinate spaces at $\sigma_a\to 0$ and 
$\sigma_a\to\infty$ respectively:
\begin{equation}
\phi^\sigma(\q,\p_a)\underset{\sigma\to 0}{\longmapsto}
(2\pi)^3\,2E_{\rm q}\delta_3(\q-\p_a),\;\quad 
\phi^\sigma(\q,\p_a)\underset{\sigma\to\infty}{\longmapsto} { N}_\infty. 
\label{45_1}
\end{equation} 
Furthermore, because $2(qp_a)=2m^2_a-(q-p_a)^2$, the nonrelativistic limit of 
(\ref{41}), (\ref{45}) for $g_1=\sigma^{-2}_a$ exactly reproduces\cite{k_t2} 
the usual Gaussian profiles in coordinate and momentum representations  
independently of normalization arbitrariness (\ref{41_0}). 
%%%%%%%%%%%%%%%%%%%

To get most simple example preserving oscillations we take an overlap function with 
one wave packet (\ref{41}), (\ref{45}) in the one vertex representing the another vertex 
by infinitely heavy nucleus in the spirit of Kobzarev model\cite{kbz}. For antineutrino 
this may be achieved for example in the Grimus-Stockinger model\cite{Grimus} with 
nucleons in creation vertex $\{C\}$ fixed in infinitely heavy nucleus and the bound state 
of initial electron in detection vertex $\{D\}$ replaced by above wave packet state for 
free electron. 
For the case of neutrino one can take the decay $\pi^+\to\mu^++\nu_\mu$ with wave packet 
(\ref{41}) for pion state and plane wave for muon in $\{C\}$- vertex taking $\{D\}$- 
vertex as in the Kobzarev model.\cite{kbz,Beuthe} The overlap functions\cite{Beuthe} 
${\cal F}_{C,D}(-p)$ for the first case (compare with Eq. (11) from 
Ref. 4 %%\cite{Grimus} 
for $\widetilde{J}_\lambda=$ const), with 
$E_{\rm w}=\sqrt{\vec{\rm w}^2+m^2_e}$, are reduced to: 
%%\begin{equation}
\begin{eqnarray}
&&\!\!\!\!\!\!\!\!\!\!\!\!\!\!\!\!\!\!\!
{\cal F}_{C}(-p)={\rm const}\;\delta(p^0-Q^0_C), \quad Q^0_C=\Delta M-U^0,  
\quad Q_D=K-W, 
\label{46}  \\
%% \end{equation}
%% \begin{equation}
&&\!\!\!\!\!\!\!\!\!\!\!\!\!\!\!\!\!\!\! 
{\cal F}_{D}(-p)=(2\pi)^4\delta\left(p^0+E_{\rm w}-K^0\right)
\frac{\phi^\sigma_e(\vec{\rm w},\vec{\rm W})}{2E_{\rm w}}, \quad 
\vec{\rm w}=\vec{\rm K}-\p, \quad W^0=E_{\rm W}, 
\label{47}
%% \end{equation}
\end{eqnarray}
where: $\Delta M$ is neutron-proton mass difference, $U^\mu=(U^0,\vec{\rm U})$ is 
4-momentum of created electron in vertex $\{C\}$; $W^\mu=(W^0,\vec{\rm W})$ is the 
most probable 4-mo\-men\-tum of incoming initial electron with 
$\zeta_e(m_e,\sigma_e)=g_1W$ in Eq. (\ref{41}), $w^\mu=(E_{\rm w},\vec{\rm w})$, and 
$K^\mu=(K^0,\vec{\rm K})$ is the total outgoing 4-momentum carried away by final electron 
with final antineutrino in vertex $\{D\}$. Note that 3-vectors $\vec{\rm U}$ and 
$\vec{\rm Q}_C$ become indefinite for this case. 
The product of these functions belongs to $S(\vec{\rm R}^3_\p)$ only for replaced 
argument of second delta-function with $E_{\rm w}\mapsto E_{\rm W}$, which appeares 
naturally in sharply peak approximation for 
$|\p-\vec{\rm Q}_D|=|\vec{\rm W}-\vec{\rm w}|\ll m_e$, leading to:  
\begin{eqnarray}
&&\!\!\!\!\!\!\!\!\!\!\!\!\!\!\!\!\!\!\!
{\cal F}(-p)\approx {\rm const}\;\delta(p^0-Q^0_C)\,\delta(p^0-Q^0_D)
\frac{N_{\sigma_e}}{2W^0}e^{-g_1(wW)},\;\mbox{ where: }\; g_1=\frac{1}{\sigma^2_e}, 
\label{48} \\
&&\!\!\!\!\!\!\!\!\!\!\!\!\!\!\!\!\!\!\!
(wW)=E_{\vec{\rm w}}W^0-(\vec{\rm w}\!\cdot\!\vec{\rm W})\approx
m^2_e+\frac 12 \left((\p-\vec{\rm Q}_D)^j(\delta^{jl}-{\rm V}^j{\rm V}^l)
(\p-\vec{\rm Q}_D)^l\right)=
\label{49_0} \\
&&\!\!\!\!\!\!\!\!\!\!\!\!\!\!\!\!\!\!\!
=m^2_e+\frac 12 \left\{\vec{\rm Q}^2_D-(\vec{\rm Q}_D\!\cdot\!\vec{\rm V})^2+\p^2-
2\left(\p\!\cdot\!\left[\vec{\rm Q}_D-\vec{\rm V}(\vec{\rm Q}_D\!\cdot\!\vec{\rm V})
\right]\right)-\vec{\rm V}^2(\p\vec{\rm B}\p) \right\}, 
\label{49} \\
&&\!\!\!\!\!\!\!\!\!\!\!\!\!\!\!\!\!\!\!
\mbox{for: }\; 
\vec{\rm V}=\frac{\vec{\rm W}}{W^0}=|\vec{\rm V}|\vec{\omega},\quad \vec{\omega}^2=1, 
\quad {\rm B}^{jl}=\omega^j\omega^l, \quad 
{\rm B}^{jl}_0=\omega^j\omega^l-\frac{\delta^{il}}3,
\label{49_1}
\end{eqnarray}
with vector $\vec{\rm V}$ as a most probable velocity of initial electron in detector. 
To define the coefficients $\Upsilon_s(k,\n)$ (\ref{39_0}) of asymptotic expansion 
(\ref{38}) by making use of Eqs. (\ref{29})--(\ref{30_0}),  only the last two summands in 
the line (\ref{49}) are now necessary, giving exactly the discussed here structure of the 
function $\Phi(\q)$ (where due to the absence of vector $\vec{\rm Q}_C$ its role 
partially plays vector $\vec{\rm V}$). 
For non-relativistic initial electron $|\vec{\rm V}|\ll 1$, and neglecting this velocity,   
for $\vec{\rm Q}_D=|\vec{\rm Q}_D|\rr$ and real $\lambda=g_1 k|\vec{\rm Q}_D|>0$, 
now it is enough to take (see Appendix B for details and another cases): 
\begin{eqnarray}
&&\!\!\!\!\!\!\!\!\!\!\!\!\!\!\!\!\!\!\!\!\!
\Phi(\q)=\Phi(-\p)=\Phi(-k\n)\underset{|\vec{\rm V}|\ll 1}{\longmapsto} 
e^{\lambda\xi},\quad \xi=(\rr\cdot\n)=\cos\Theta,
\;\mbox{ whence:}
\label{42} \\
&&\!\!\!\!\!\!\!\!\!\!\!\!\!\!\!\!\!\!\!\!\! 
C_1(k,\n) = 2\lambda \xi-\lambda^2(1-\xi^2), \quad 
\Upsilon_1(k,\n)= -4\lambda^2\left[1-(2+\lambda\xi)(1-\xi^2)\right]. 
\label{43}
\end{eqnarray} 
So $\Upsilon_1<0$, if $\lambda\xi<\left[(1-\xi^2)^{-1}-2\right]$, that 
$\forall\,\lambda>0$ includes interval $3\pi/4\leq\Theta\leq \pi$,   
\begin{equation}
\mbox{which for }\;\lambda\to\infty,\;\mbox{ dilates to: }\; 
0\leq \Theta \leq 1/\sqrt{\lambda},\quad 
\pi/2+1/\lambda  \leq \Theta \leq \pi. 
\label{44_1} 
\end{equation}
With the same function (\ref{42}) $\forall\,\lambda<0$ this condition implies interval 
$0\leq\Theta\leq \pi/4$,
\begin{equation}
\mbox{dilating to: }\; 
0\leq\Theta\leq \pi/2-1/|\lambda|,\quad \pi-1/\sqrt{|\lambda|}\leq\Theta\leq \pi, \;
\mbox{ for }\;|\lambda|\to\infty,  
\label{44_0} 
\end{equation}
that means the formal interchange in (\ref{44_1}) of forward and backward hemispheres   
$\Theta\mapsto\pi-\Theta$. Similarly for the above mentioned case with neutrino, one can 
obtain for $\Phi(\q)=\Phi(k\n)$ the same result (\ref{42}) with $\lambda>0$. 
Therefore for the both cases $\Upsilon_1(k,\n)<0$ near the forward direction (\ref{44_1}) 
with the narrow wave packet $\sigma_{e,\pi}\to 0$ and/or with the high energy $p^0=Q^0_D$ 
of (anti-) neutrino. Since due to Eq. (\ref{49_0}), only this region in fact contributes 
to the averaging relative to $\rr$, from (\ref{43}) for the parameter $\varrho_0$ 
(\ref{41_01}), with $k=\sqrt{(Q^0_D)^2-m^2_j}\approx|\vec{\rm Q}_D|$, and $Q^0_D=Q^0_C$ 
(\ref{46}), one finds: 
\begin{eqnarray} 
&&\!\!\!\!\!\!\!\!\!\!\!\!\!\!\!\!\!\!\!\!\!
\overline{\Upsilon}_1 \simeq -4\lambda^2, \quad 
\varrho^2_0\simeq\frac{4\lambda^2}{(2k)^2}=\frac{\vec{\rm Q}^2_D}{\sigma^4_e}, \quad 
\varrho_0\approx \frac{k}{\sigma^2_e}, 
\label{50} \\
&&\!\!\!\!\!\!\!\!\!\!\!\!\!\!\!\!\!\!\!\!\!
\mbox{where: }\;
m_j\leq \sigma_e \ll m_e,\quad k\leq Q^0_C\leq \Delta M-m_e.
\label{51} 
\end{eqnarray} 
For $k\simeq\Delta M-m_e\approx 0.783\,\mbox{MeV}$ anf 
$\sigma_e\sim m_j\simeq 0,2\,\mbox{eV}$ it follows $\varrho_0\simeq 3,86\,\mbox{m}$. 
Note that for real $\beta$- decay of heavy nucleus the value of antineutrino energy 
$Q^0_C\simeq k$ may be in many times greater. 
This qualitative analytical estimation illustrates and elucidates the numerical 
calculations\cite{N_shk,NN_shk} with more detailed models, that give the best 
value for this parameter $\varrho_0\simeq 3,5\,\mbox{m}$. 
The obtained analytical results (\ref{29})--(\ref{30_0}), (\ref{38}), (\ref{39_0}), 
together with relations from Appendix B can essentially simplify such calculations 
making them probably more transparent and model-independent. 

%%%%%%%%%%%%%%%%%%%%%%%%%%%%%%%%%%%%%%%%%%%%%%%%%%%%%%%%
\section{Conclusions}
%%%%%%%%%%%%%%%%%%%%%%%%%%%%%%%%%%%%%%%%%%%%%%%%%%%%%%%
As shown in Appendix A, following the same way as in Lemma 1 the operator 
expansion (\ref{18}) for the Green function of Helmholtz equation may be generalized for 
arbitrary $N$- dimensional Euclidean space. 
The same concerns the assertion of Theorem 1 about asymptotic expansion for 
$N$- dimensional analog of integral (\ref{1}). 

The following comments are in order. 
Since the main principles of quantum field theory\cite{blt} require the wave packets 
belong to $S(\vec{\rm R}^3_{\q})$, the same holds true for their convolution and thus 
for $\Phi(\q)$, as was supposed in Theorem 1. This property ensures the possibility 
of further asymptotic expansion (\ref{29}), (\ref{33}) of integral (\ref{1}) and defines 
the order of leading corrections to Grimus-Stockinger asymptotic,\cite{Grimus} which is 
crucial for explanation\cite{N_shk,NN_shk} (\ref{41_01}) of (anti-) neutrino deficit. 

For electron in stationary bound state\cite{Grimus} with energy $W^0=E_b$ instead of the 
wave packet state the overlap function (\ref{47}), with the also indefinite now vectors 
$\vec{\rm W}$, $\vec{\rm K}$ and $\vec{\rm Q}_D$, becomes:  
${\cal F}_{D}(-p)=\delta\left(p^0+E_{b}-K^0\right)\widetilde{\psi}_b(\vec{\rm w})
\in C^\infty(\vec{\rm R}^3_{\p})$, where  
$|\widetilde{\psi}_b(\vec{\rm K}-\p)|\leq O(|\p|^{-4})$ for $|\p|\to\infty$. 
It may be argued following to the Lemma 1 of Ref. 2 and the Theorem of Ref. 4, that this 
fourth power fall-off should be enough for the well definiteness of the asymptotic 
expansion (\ref{29}) up to the term with $s=2$, and therefore the expansion (\ref{38}),  
(\ref{39_0}) up to the term with $n=1$, what preserves Eqs. (\ref{41_00}), (\ref{41_01}). 
On the other hand the infinite coherence length of neutrino oscillations for the 
stationary case\cite {Grimus} is also the price for the absence of external wave 
packets like (\ref{41}), (\ref{45}); only their presence in both the creation $\{C\}$ 
and detection $\{D\}$ vertices provides a finite value of coherence 
length.\cite{Beuthe,NN} So we conclude, that again as in the Ref. 4:  %%\cite {Grimus}   
``Therefore the physics under consideration seems to comply naturally with the 
mathematical requirements''.

Since the wave packet (\ref{41}) in momentum representation 
$\phi^\sigma(\q,\p_a)\in S(\vec{\rm R}^3_{\q})$, the same takes place\cite{blt} for its 
coordinate representation (\ref{45}). 
So, for fixed $t$, $F_{p_ax_a}(t,\xv)\in S(\vec{\rm R}^3_{\xv})$ obviously with infinite 
support and the same holds true for their convolution $\phi(\xv)$ in Eqs. (\ref{32_1}) 
and (\ref{12}). Therefore the expansion (\ref{33}) acquires also the additional 
asymptotic nature, with nonzero corrections (\ref{34}), governed by detailed form of 
asymptotic behavior of the function $\phi(\xv)$ itself. These corrections inevitable  
introduce an additional dimensional parameter and may also affect the behavior of 
${\cal J}(\vec{\rm R})$ at microscopically big but macroscopically small 
distances $R$, discussed here and in Ref. 2,3,9. %%%\cite{N_shk,NN_shk,anom}

%%%%%%%%%%%%%%%%%%%%%%%%%%%%%%%%%%%%%%%%%%%%%%%%%%%%%%%%%%%%%%%%%%%%%%%
%%% \begin{equation}
%% \end{equation}
%%%%%%%%%%%%%%%%%%%%%%%%%%%%%%%%%%%%%%%%%%%%%%%%%%%%%%%
%%% (\ref{})
%%% distinguishes
%%%%%%%%%%%%%%%%%%%%%%%%%%%%%%%%%%%%%%%%%%%%%%%%%%%%%%%
\appendix 
\section{}  %%% {Appendix A}

For $N$-dimension space, defining: $\nu =(N-2)/2$, $a=(N-3)/2=\nu-\frac 12$,    
$l=0,1,2,\ldots$, $j=l+a$, with: %%% $l+\nu=j+\frac 12$ and 
$l(l+2\nu)+\nu^2-\frac 14=l(l+2\nu)+a(a+1)=j(j+1)$, for $\vec{\rm R}=R\n$, $\xv=r\vs$, 
$\ka=k\n$, one has\cite{b_e2,vil} Gegenbauer polinomials $C^{(\nu)}_l((\n\cdot\vs))$ as 
eigenfunction of $N$-dimension square angular momenta operator ${\cal L}^{(N)}_{\n}\!,$ 
where for $\vec{\nabla}_{\vec{\rm R}}$ in (\ref{3}), instead of (\ref{5}), (\ref{6}): 
$[(\vec{\partial}_{\n})_\beta,n_\alpha]=\delta_{\alpha\beta}-n_\alpha n_\beta$, 
$(\n\cdot\vec{\partial}_{\n})=0$, 
$(\vec{\partial}_{\n}\cdot\n)=N-1$, 
$x_\alpha\partial_\beta-x_\beta\partial_\alpha=
n_\alpha(\vec{\partial}_{\n})_\beta-n_\beta(\vec{\partial}_{\n})_\alpha=
i{\rm L}_{\alpha\beta}$ and 
$-\vec{\partial}^2_{\n}=\vec{\rm L}^2/2={\cal L}^{(N)}_{\n}$, so that:\cite{vil}   
\begin{eqnarray}
&&\!\!\!\!\!\!\!\!\!\!\!\!\!\!\!\!\!\!\!\!\!\!\!\! 
\vec{\nabla}^2_{\vec{\rm R}}=
\partial^2_R+\frac{N-1}R \partial_R-\frac{{\cal L}^{(N)}_{\n}}{R^2}=
\frac 1{R^{a+1}}\partial^2_R R^{a+1}-\frac{a(a+1)}{R^2}-\frac{{\cal L}^{(N)}_{\n}}{R^2}, 
\label{A_2} \\
&&\!\!\!\!\!\!\!\!\!\!\!\!\!\!\!\!\!\!\!\!\!\!\!\! 
{\cal L}^{(N)}_{\n}C^{(\nu)}_l((\n\cdot\vs))=l(l+2\nu)C^{(\nu)}_l((\n\cdot\vs)).
\label{A_3}
\end{eqnarray}
In definitions (\ref{24}), (\ref{24_1}), where,\cite{olv,T,b_e2,vil} for $l\mapsto j$, 
with $|{\rm arg}\,z-\beta_{1,2}|<\pi/2$:  
\begin{eqnarray}
&&\!\!\!\!\!\!\!\!\!\!\!\!\!\!\!\!\!\!\!\!
K_\lambda(z)=\frac 12 \int\limits^{\infty e^{-i\beta_1}}_{0 e^{i\beta_2}}
\frac{dt}{t}t^{\pm\lambda}\exp\left\{-\,\frac z2\left(t+\frac 1t\right)\right\}, 
\label{A_1} \\
&&\!\!\!\!\!\!\!\!\!\!\!\!\!\!\!\!\!\!\!\!
\mbox{and: }\; 
\int\limits^\infty_0\! dr\, \psi_{j\,0}(kr)\psi_{j\,0}(qr)=\frac{\pi}2\delta(q-k), 
\label{24_0} 
\end{eqnarray}
according to formulas (8.534), (8.532) of Ref. 12 %%\cite{b_e2} 
($H^{(1),(2)}_\nu(z)$ are the Hankel functions,\cite{olv,b_e2}) the plane wave and the 
Green function now read:\cite{vil,k_p} 
\begin{eqnarray}
&&\!\!\!\!\!\!\!\!\!\!\!\!\!\!\!\! 
e^{\mp i(\ka\cdot\xv)}=\sqrt{\frac 2\pi}\,
\frac{2^{\nu-1}\Gamma(\nu)}{(\mp ikr)^a kr}
\sum^\infty_{j=a}(2j+1)i^{\mp j}\,\psi_{j\,0}(kr)C^{(\nu)}_l((\n\cdot\vs)),
\label{A_4_0} \\
&&\!\!\!\!\!\!\!\!\!\!\!\!\!\!\!\! 
\left\langle\vec{\rm R}\biggr|\frac{1}{-\vec{\nabla}^2-k^2\mp i0}\biggl|\xv\right\rangle=
\!\int\!\frac{d^N{\rm q}}{(2\pi)^N}\frac{e^{i(\q\cdot(\vec{\rm R}-\xv))}}
{(\q^2-k^2\mp i0)}=
\label{A_4} \\
&&\!\!\!\!\!\!\!\!\!\!\!\!\!\!\!\! 
=\pm \frac{i}{4}\,\left(\frac{k}{2\pi}\right)^{\nu}
\frac{H_\nu^{\overset{(1)}{(2)}}%%%%{\stackrel{(1)}{(2)}}
(k|\vec{\rm R}-\xv|)}{|\vec{\rm R}-\xv|^\nu}=
\sqrt{\frac \pi 2}\,\frac{(\mp ik)^a}{(2\pi)^{\nu+1}}\,
\frac{\chi_a\left(\mp ik|\vec{\rm R}-\xv|\right)} {|\vec{\rm R}-\xv|^{a+1}} = 
\label{A_5}\\
&&\!\!\!\!\!\!\!\!\!\!\!\!\!\!\!\! 
=\frac{\Gamma(\nu)}{4\pi^{\nu+1}}\,\frac{1}{(Rr)^{a+1}}
\sum^\infty_{j=a}(2j+1)\,\frac{i^{\mp j}}k\,\chi_j(\mp ikR)\,
\psi_{j\,0}(kr)C^{(\nu)}_l((\n\cdot\vs)), 
\label{A_6}
\end{eqnarray}
for $R>r$ and $l=j-a$,  
where $a$ and $j\geq a$ are both integer for odd $N$ or half integer for even $N$. 
So, for $R>r$, with again really not appearing operator $l\mapsto\Lambda^{(N)}_{\n}$, 
one obtains similarly the following generalization of Lemma 1 (\ref{18_0}), (\ref{18}):
\begin{eqnarray}
&&\!\!\!\!\!\!\!\!\!\!\!\!\!\!\!\!\!\!\!\!\!  
\sqrt{\frac \pi 2}\, \frac{(\mp ik)^a}{(2\pi)^{\nu+1}}\,
\frac{\chi_a\left(\mp ik|\vec{\rm R}-\xv|\right)} {|\vec{\rm R}-\xv|^{a+1}} =
\sqrt{\frac \pi 2}\, \frac{(\mp ik)^a}{(2\pi)^{\nu+1}}\,
\frac{\chi_{\Lambda^{(N)}_{\n}+a}(\mp ikR)}{R^{a+1}}\,e^{\mp ik(\n\cdot\xv)}\sim 
\label{A_7} \\
&&\!\!\!\!\!\!\!\!\!\!\!\!\!\!\!\!\!\!\!\!\!    
\sim \sqrt{\frac \pi 2}\, \frac{(\mp ik)^a}{(2\pi)^{\nu+1}}
\frac{e^{\pm ikR}}{R^{a+1}}\!
\left\{1+\sum^\infty_{s=1}\frac{ \displaystyle 
\prod\limits^s_{\mu=1}\!\left[{\cal L}^{(N)}_{\n}+a(a+1)-\mu(\mu-1)\right]}
{s!\,(\mp 2ikR)^s}\right\}\! e^{\mp ik(\n\cdot\xv)}\!. 
\label{A_8} 
\end{eqnarray}
Here (\ref{A_7}) is again an operator rewrite of (\ref{A_6}) via (\ref{A_4_0}) and 
positively defined self-adjoint operator with positive eigenvalue, from which  
$\chi_j(z)$ (\ref{24}) depends on again as a whole function (\ref{A_1}): 
\begin{equation}
{\cal L}^{(N)}_{\n}+a(a+1)+\frac 14=\left(\Lambda^{(N)}_{\n}+a+\frac 12\right)^2
\longmapsto \left(j+\frac 12\right)^2,  
\label{A_9}
\end{equation}
and again the asymptotic version\cite{olv,b_e2} of expansion (\ref{24}) is used. 
Due to (\ref{A_2}), instead of operator ${\cal L}^{(N)}_{\n}+a(a+1)$ in (\ref{A_8}),  
again may be used the operator from the square brackets of the radial Schr\"odinger 
equation:  
\begin{eqnarray}
&&\!\!\!\!\!\!\!\!\!\!\!\!\!\!\!\!\!\!\!\!\!\!\!\! 
\left[r^2\left(\frac 1{r^{a+1}}\,\partial^2_r\,r^{a+1}+k^2\right)\right]
\frac{\psi_{j\,0}(kr)}{r^{a+1}}=j(j+1)\frac{\psi_{j\,0}(kr)}{r^{a+1}}. 
\label{A_10}
\end{eqnarray}
%%%%%%%%%%%%%%%%%%%%%%%%%%%%%%%%%%%%%%%%%%%%%%%%%%%%%%%%%%%%%%%
\section{} %%% {Appendix B}
%%%%%%%%%%%%%%%%%%%%%%%%%%%%%%%%%%%%%%%%%%%%%%%%%%%%%%%%%%%%%%%%
The differential operators $\vec{\partial}_{\n}$ and 
${\cal L}_{\n}=-\vec{\partial}^2_{\n}$ defined in Eqs. (\ref{3})--(\ref{7}) 
have the following properties for any differentiable scalar functions $f(\xi)$, 
${\cal W}(\xi,\varsigma)$, ${\cal H}(\zeta)$ with $\xi$ from Eq. (\ref{42}), 
$\varsigma=(\vec{\varrho}\cdot\n)$, $\zeta_m=(\n\vec{\rm B}^m\n)$, $\zeta\equiv\zeta_1$, 
$Tr\{\vec{\rm B}^m\}=3\overline{\zeta}_m$, $\zeta_{m0}=\zeta_m-\overline{\zeta}_m$: 
\begin{eqnarray}
&&\!\!\!\!\!\!\!\!\!\!\!\!\!\!\!\!\!\!\!\!\!\!\!\! 
\vec{\partial}_{\n}f(\xi)=
\vec{\partial}_{\n}f((\rr\cdot\n))=\rr_\perp\partial_\xi f(\xi)+f(\xi)\vec{\partial}_{\n}, 
\quad \rr_\perp=\rr-\n(\rr\cdot\n),  
\label{B_1} \\
&&\!\!\!\!\!\!\!\!\!\!\!\!\!\!\!\!\!\!\!\!\!\!\!\! 
(\vec{\partial}_{\n}\xi)=\rr_\perp,\quad 
(\vec{\partial}_{\n}\varsigma)=\vec{\varrho}_\perp, \quad 
\left((\vec{\partial}_{\n})^j\zeta_m\right)=
2\left[(\vec{\rm B}^m\n)^j-{\rm n}^j\zeta_m\right],
\label{B_1_1} \\
&&\!\!\!\!\!\!\!\!\!\!\!\!\!\!\!\!\!\!\!\!\!\!\!\! 
{\cal L}_{\n}\xi=2\xi, \quad {\cal L}_{\n}\varsigma=2\varsigma,\quad 
{\cal L}_{\n}\zeta_m=2\left[3\zeta_m-Tr\{\vec{\rm B}^m\}\right]\equiv 6\zeta_{m0}=
{\cal L}_{\n}\zeta_{m0},
\label{B_4} \\
&&\!\!\!\!\!\!\!\!\!\!\!\!\!\!\!\!\!\!\!\!\!\!\!\! 
f(\xi)C_1(\xi)={\cal L}_{\n}f(\xi)=\partial_{\xi}(\xi^2-1)\partial_{\xi}f(\xi)=
\left[(\xi^2-1)\partial^2_{\xi^2}+2\xi\partial_{\xi}\right]f(\xi), 
\label{B_2} \\
&&\!\!\!\!\!\!\!\!\!\!\!\!\!\!\!\!\!\!\!\!\!\!\!\! 
f(\xi)C_2(\xi)=\left[(6\xi^2-2)\partial^2_{\xi^2}+4\xi(\xi^2-1)\partial^3_{\xi^3}+
\frac 12(\xi^2-1)^2\partial^4_{\xi^4} \right]f(\xi), 
\label{B_3} \\
&&\!\!\!\!\!\!\!\!\!\!\!\!\!\!\!\!\!\!\!\!\!\!\!\! 
{\cal L}_{\n}{\cal W}(\xi,\varsigma)=\left[\partial_{\xi}(\xi^2-1)\partial_{\xi}+
\partial_{\varsigma}(\varsigma^2-1)\partial_{\varsigma}+
2\left(\xi\varsigma-(\vec{\varrho}\cdot\rr)\right)\partial^2_{\xi\varsigma}\right]
{\cal W}(\xi,\varsigma), 
\label{B_5} \\
&&\!\!\!\!\!\!\!\!\!\!\!\!\!\!\!\!\!\!\!\!\!\!\!\! 
{\cal L}_{\n}f(\xi)g(\varsigma)=
g(\varsigma){\cal L}_{\n}f(\xi)+f(\xi){\cal L}_{\n}g(\varsigma)-
2(\rr_\perp\cdot\vec{\varrho}_\perp)\partial_\xi f(\xi)\partial_\varsigma g(\varsigma), 
\label{B_5_1} \\
&&\!\!\!\!\!\!\!\!\!\!\!\!\!\!\!\!\!\!\!\!\!\!\!\! 
{\cal H}(\zeta)C_1(\zeta)=
{\cal L}_{\n}{\cal H}(\zeta)=\left\{6\zeta_0\partial_\zeta+
4(\zeta^2-\zeta_2)\partial^2_{\zeta^2}\right\}{\cal H}(\zeta), 
\quad \; \zeta_0\equiv \zeta_{10}, 
\label{B_6} \\
&&\!\!\!\!\!\!\!\!\!\!\!\!\!\!\!\!\!\!\!\!\!\!\!\! 
{\cal H}(\zeta)C_2(\zeta)=
\frac 12\left[{\cal L}^2_{\n}-2{\cal L}_{\n}\right]{\cal H}(\zeta)= 
\label{B_6_1} \\
&&\!\!\!\!\!\!\!\!\!\!\!\!\!\!\!\!\!\!\!\!\!\!\!\! 
=\left\{12\zeta_0\partial_{\zeta}+
6\left[(4\zeta+3\zeta_0)\zeta_0+6(\zeta^2-\zeta_2)-2\zeta_{20}\right]
\partial^2_{\zeta^2}+ \right.
\nonumber \\
&&\!\!\!\!\!\!\!\!\!\!\!\!\!\!\!\!\!\!\!\!\!\!\!\! 
\left. +8\left[(4\zeta+3\zeta_0)(\zeta^2-\zeta_2)+2(\zeta_3-\zeta\zeta_2)\right]
\partial^3_{\zeta^3}+8(\zeta^2-\zeta_2)^2\partial^4_{\zeta^4}\right\}{\cal H}(\zeta).
\label{B_7}
\end{eqnarray}
When $f(\xi)\mapsto e^{\lambda\xi}$, Eqs. (\ref{B_2}), (\ref{B_3}) with 
$\partial_{\xi}\mapsto\lambda$ immediately give Eqs. (\ref{43}). 
Similarly for ${\cal H}(\zeta)\mapsto e^{\lambda\zeta}$ one has 
$\partial_{\zeta}\mapsto\lambda$ in Eqs. (\ref{B_6})--(\ref{B_7}), leading to:. 
\begin{eqnarray}
&&\!\!\!\!\!\!\!\!\!\!\!\!\!\!\!\!\!\!\!\!\!\!\!\! 
\Upsilon_1=-8\lambda\left[a\lambda^2+b\lambda+c\right],\;\mbox{ with: }\;
a=4\left[2\zeta^3+\zeta_3-3\zeta\zeta_2\right], 
\label{B_8} \\
&&\!\!\!\!\!\!\!\!\!\!\!\!\!\!\!\!\!\!\!\!\!\!\!\! 
b=3\left[2\zeta_0\zeta+3(\zeta^2-\zeta_2)-\zeta_{20}\right],\quad 
c=3\zeta_0,\quad \zeta_m=(\n\vec{\rm B}^\sigma\vec{\rm B}^{m-\sigma}\n)>0,
\label{B_9} \\
&&\!\!\!\!\!\!\!\!\!\!\!\!\!\!\!\!\!\!\!\!\!\!\!\! 
(\zeta_m)^2\leq \zeta_{2\sigma}\zeta_{2m-2\sigma}, 
\quad 
\zeta_m\geq\zeta\zeta_{m-1}\geq\zeta^2\zeta_{m-2}\geq\ldots\geq\zeta^k\zeta_{m-k}
\ldots\geq\zeta^m, 
\label{B_10}
\end{eqnarray}
due to positive definiteness of $\vec{\rm B}$ and Cauchy-Schwarz inequality. 
Formally the case (\ref{B_8}), (\ref{B_9}) appeares in Eqs. (\ref{49}) for 
``collinear'' events with $\vec{\rm K} \parallel \vec{\rm V}$ only, for  
$\lambda=g_1 k^2 \vec{\rm V}^2/2>0$ and degenerate tensor $\vec{\rm B}$ (\ref{49_1}) 
evidently realizing all inequalities (\ref{B_10}). 
The structure of $\vec{\rm B}$ and $\lambda$ becomes much more complicated for the 
complete models.\cite{N_shk,NN} 
Nevertheless the most general case containing all variables simultaneously, like  
$\Phi(\mp k\n)\mapsto e^{\alpha\xi}e^{\beta\varsigma}e^{\lambda\zeta}$, may be 
considered by the same way using the above relations together.  

%%%%%%%%%%%%%%%%%%%%%%%%%%%%%%%%%%%%%%%%%%%%%%%%%%%%%%%%%%%%%%%%%
%%%%%%%%%%%%%%%%%%%%%%%%%%%%%%%%%%%%%%%%%%%%%%%%%%%%%%%%%%%%%%%%
\section*{Acknowledgments}
Authors thank V. Naumov, D. Naumov, and N. Ilyin for useful discussions and Reviewer 
for constructive comments.

\end{document}